\documentclass[letterpaper, 10 pt, conference]{ieeeconf}  

\IEEEoverridecommandlockouts                              
  \usepackage{amsmath,amsfonts}
\usepackage{amssymb}
\usepackage{algorithmic}
\usepackage{algorithm}
\usepackage{array}
\usepackage[caption=false,font=normalsize,labelfont=sf,textfont=sf]{subfig}
\usepackage{textcomp}
\usepackage{stfloats}
\usepackage{url}
\usepackage{verbatim}
\usepackage{graphicx}
\usepackage{cite}
\usepackage{color}
\usepackage{float}

\usepackage{tikz}

\usepackage[T1]{fontenc}

\graphicspath{ {fig} }
\newtheorem{definition}{\textit{\bf Definition}}
\newtheorem{remark}{\textit{\bf Remark}}

\newtheorem{lemma}{\textit{\bf Lemma}}

\newtheorem{theorem}{\textbf{\textit{Theorem}}}

\newtheorem{problem}{\textit{\bf Problem}}
\def\BibTeX{{\rm B\kern-.05em{\sc i\kern-.025em b}\kern-.08em
    T\kern-.1667em\lower.7ex\hbox{E}\kern-.125emX}}                                                        

\title{\LARGE \bf
Traffic Characterization of Event-Triggered Control Systems: \\A Geometric-Algebraic Perspective
\thanks{
	\textcopyright~2026 IEEE. Personal use of this material is permitted. Permission from IEEE must be obtained for all other uses, in any current or future media, including reprinting/republishing this material for advertising or promotional purposes, creating new collective works, for resale or redistribution to servers or lists, or reuse of any copyrighted component of this work in other works.}
}

\author{Tao Chen, Hongju Wang and Wenfeng Hu$^{*}$
\thanks{*This work was supported in part by the National Natural Science Foundation of China under Grant 62273358; in part by the Natural Science Foundation of Hunan Province under Grant 2023JJ20075. \textit{(Corresponding author: Wenfeng Hu.)}}
\thanks{Tao Chen, Hongju Wang and Wenfeng Hu are with the School of Automation, Central South University, Changsha,  410000, China(e-mail: tchen@csu.edu.cn; wenfenghu@csu.edu.cn; wanghongju@csu.edu.cn).
        }%
}

\begin{document}

\maketitle
\begin{tikzpicture}[remember picture,overlay]
	\node[anchor=north, yshift=-30pt] at (current page.north)
	{\footnotesize\textit{Accepted version of a paper accepted by the 2026 American Control Conference (ACC 2026)}};
\end{tikzpicture}

\thispagestyle{empty}
\pagestyle{empty}

\begin{abstract}
This paper characterizes the triggering behaviors of event-triggered control systems from a geometric–algebraic perspective. We first model the feasibility of inter-event time transition relations as a nonconvex quadratic constraint satisfaction problem and reformulate it as an equivalent linear cone problem, which provides a clearer geometric description of the feasible region, making subsequent analysis more reliable. Building on this formulation, we establish necessary and sufficient conditions that rigorously determine whether a given transition relation is feasible. Based on this condition, we propose an algorithm that computes the set of all feasible transition relations. Numerical simulations further demonstrate how the feasibility of specific transitions evolves with the control parameter~$\sigma$, with visualizations of the feasible state space offering intuitive insight into parameter selection and system design.
\end{abstract}

\section{INTRODUCTION}

    In modern networked control systems (NCSs), event-triggered control (ETC) \cite{Tabuada2007Event,Dimos2012Distributed, Heemels2013Periodic} has become a key paradigm for reducing communication burden while ensuring control performance. Unlike periodic sampling, ETC transmits information only when required, thereby achieving substantial savings in network usage. However, this advantage comes with an inherent trade-off. A loose triggering condition reduces communication but may lead to performance degradation, while a tight condition preserves performance but risks overloading the network \cite{chen2020often}. This essential trade-off makes it essential to analyze the communication traffic characterization generated by ETC, as it provides deeper insight into the network usage in such systems.

    Existing research has made substantial progress in designing triggering conditions with two main objectives. One is guaranteeing stability, such as asymptotic or input-to-state stability. The other is ensuring a strictly positive minimum inter-event time (MIET) to exclude Zeno behavior \cite{girard2014dynamic, Hu2020DistributedEvent}. These studies reveal the fundamental trade-off between control performance and communication resources, treating the triggering condition as the main design object. The focus of this paper is different. We study the triggering behaviors themselves, characterized by inter-event times (IETs) and their transition relations, such as $\{t_1 \to t_2, t_1 \to t_5, t_2 \to t_9,\cdots \}$, which directly determine the traffic characterization in event-triggered control systems.

	  Existing studies on triggering behaviors can be divided into two categories. Some works focus on IETs, analyzing their steady-state behavior\cite{Heemels2023ExplainingMystery}. Other works go further to study transition relations among IETs, which describe all possible triggering behaviors. Seminal work in this direction has been pursued through abstraction-based methods. Manuel et al.~\cite{Manuel2018FormalTrafficLTI} developed finite-state abstractions to model the triggering behavior of linear event-triggered systems, later extended to output-feedback~\cite{Manuel2018TrafficPertubate}, nonlinear~\cite{Manuel2022TrafficNonlinear}, and stochastic settings~\cite{Manuel2023FormalStochastic}. These abstractions construct traffic models that represent feasible IETs and their transition relations as a finite-state transition system. While such models have proven useful, they almost always rely on approximations or relaxed numerical procedures during modeling or computation. Consequently, the resulting transition relations may include behaviors that do not actually exist in the original system. More importantly, these methods cannot explain why some transition relations occur, whereas others do not.
	
	To address this gap, this paper focuses on characterizing triggering behaviors of event-triggered systems from a geometric–algebraic perspective. The triggering behaviors are modeled as a nonconvex quadratic constraint satisfaction problem, where the feasibility of a transition relation $t_i \to t_j$ is constrained by quadratic inequalities on the initial state. Then, we convert the original nonconvex quadratic problem into an equivalent linear cone problem, which provides a geometric formulation of the feasible region and prepares the ground for further analysis.
    Based on this cone formulation, we derive an equivalent algebraic condition that converts the geometric condition to an algebraic condition for determining whether a given transition relation $t_i \to t_j$ can occur. Then, based on this condition, we propose an algorithm that determines the set of all feasible transition relations.
    Finally, we validate the method through numerical simulations, where we characterize and visualize how the control parameter $\sigma$ influences triggering behaviors by analyzing the evolution of the feasible state space.
	
\section{Preliminaries And Problem Statement}

 {\bf{Notations}}:  
     Let $\mathbb{R}$, $\mathbb{R}^{+}$, $\mathbb{N}$ denote the sets of real numbers, strictly positive real numbers and natural numbers, respectively. Let $I$ denote the identity matrix. The matrix norm is denoted by $\| \cdot \|$.
     For a vector $v \in \mathbb{R}^n$ and a square matrix $A \in \mathbb{R}^{n \times n}$, $\operatorname{diag}(v)$ denotes the $n \times n$ diagonal matrix with $v$ on its main diagonal and $\operatorname{diag}(A)$ denotes the vector in $\mathbb{R}^n$ whose $i$-th element is $A_{ii}$.

     For any vectors $x, y \in \mathbb{R}^n$, we define two operations. The inner product is defined as $\langle x, y \rangle = x^\top y$. The Hadamard product is denoted by $\circ$ and defined as $x \circ y = [x_1y_1,\dots,x_ny_n]^\top \in \mathbb{R}^n$, where $x,y \in \mathbb{R}^n$.

     The undirected Laplacian matrix is defined as $L$, which is symmetric and positive semi-definite. Its eigenvalues are real and can be ordered as $0 = \lambda_1 \leq \lambda_2 \leq \cdots \leq \lambda_n$. Then the matrix $L$ admits an eigenvalue decomposition $L = V \Lambda V^\top$, where $V$ is an orthogonal matrix satisfying $V^\top V = VV^\top = I_n$ and $\Lambda = \operatorname{diag}(\lambda_1, \lambda_2, \dots, \lambda_n)$.

	\subsection{Some Definitions}	
	\begin{definition}[Dual Cone and Polar Cone \cite{Rockafellar1970}]
		Let $E \subseteq \mathbb{R}^n$ be a cone. The \emph{dual cone} of $E$ is defined as
		$E^* := \left\{ y \in \mathbb{R}^n \,\middle|\, \langle y, x \rangle \geq 0,\ \forall x \in E \right\},$
		and the \emph{polar cone} of $E$ is defined as $E^\circ := \left\{ y \in \mathbb{R}^n \,\middle|\, \langle y, x \rangle \leq 0,\ \forall x \in E \right\}.$
	\end{definition}

	\begin{definition}[Interior, Closure, and Boundary \cite{Rockafellar1970}]
		Let $S \subseteq \mathbb{R}^n$ be a set. The \emph{interior} of $S$, denoted by $\operatorname{int}(S)$, is defined as $\operatorname{int}(S) := \left\{ x \in S \,\middle| \, \exists\, \epsilon > 0 \text{ such that } B(x, \epsilon) \subseteq S \right\}$, where $B(x, \epsilon) := \left\{ y \in \mathbb{R}^n \,\middle|\, \|y - x\| < \epsilon \right\}$ is the open ball centered at $x$ with radius $\epsilon$. The \emph{closure} of $S$, denoted by $\operatorname{cl}(S)$, is the smallest closed set containing $S$. The \emph{boundary} of $S$ is defined as
		$\partial S := \operatorname{cl}(S) \setminus \operatorname{int}(S).$
	\end{definition}

    \subsection{Problem Statement}
	Consider a linear system governed by
	\begin{align}
		\dot{x}(t) = -Lx(t), \label{eq:linearSys}
	\end{align}
	where $x(t) \in \mathbb{R}^n$ and $L \in \mathbb{R}^{n \times n}$ is symmetric and positive semi-definite. To reduce the communication burden, we adopt an ETC strategy, and the closed-loop system becomes
	\begin{align}\label{closedloopsys}
		\dot x(t) = -Lx(t_g), \quad \forall t \in [t_g, t_{g+1}),
	\end{align}
	where $t_g$ denotes the most recent triggering time, and the triggering sequence $\{t_0, t_1, \dots\} \subseteq \{h, 2h, \dots\}$ with sampling period $h > 0$.
	To determine the triggering times, define $\tau_k = t_q + kh$ for $k \in \mathbb{N}$. Then the next triggering time is
	\begin{align}\label{eq:triggering-time}
		t_{q+1} = \inf \left\{ \tau_k ~\middle|~ \phi(x(t_q), x(\tau_k), \sigma) > 0 \;\lor\; k = \bar{k} \right\},
	\end{align}
	where $\phi(\cdot)$ is the triggering function, $\sigma > 0$ is a control parameter, and $\bar{k}$ denotes the maximum allowable IET, ensuring that IETs remain bounded.
	
	Define the IET function as
	\begin{align}\label{Kx}
		\mathcal{K}(x,\sigma) := \frac{t_{q+1} - t_q}{h} \in \{1, 2, \dots, \bar{k}\}.
	\end{align}
    The triggering behavior over time is characterized by the IET sequence generated from~\eqref{Kx}, where each IET $k \in \mathcal{K}(x,\sigma)$ corresponds to an interval $kh$. To define the triggering bebaviors, we give the following definitions

    \begin{definition}[Transition Relation]
    A \emph{transition relation} $k_i \to k_j$ denotes the evolution from a current IET $k_i h$ to the next IET $k_j h$.
    \end{definition}
    \begin{definition}[Feasibility of Transition Relation]\label{Feasible}
    A transition relation $k_i \to k_j$ is said to be \emph{feasible} if $ \exists \;x(t_q) \in \mathbb{R}$ at a triggering instant $t_q$ such that $\mathcal{K}(x(t_q),\sigma)=k_i$ and $\mathcal{K}(x(t_q+k_1h),\sigma)=k_j$. Otherwise, the transition is said to be \emph{infeasible}. In this sense, the \emph{feasibility} of a triggering transition relation refers simply to whether it is classified as feasible or infeasible. 
    \end{definition}

    Given a maximum IET $\bar{k}h$, we first define the index set of all possible pairs as
    \begin{equation} \label{cl}
        \mathcal{D} = \{ (k_i, k_j) \mid k_i, k_j \in \{1,2,\dots,\bar{k}\} \}.
    \end{equation}
    Each element $(k_i,k_j) \in \mathcal{D}$ corresponds to a candidate transition relation $k_i \to k_j$. 
    Among these, the subset $\mathcal{B} \subseteq \mathcal{D}$ collects all feasible transition relations, which together characterize the triggering behaviors of the system.
	
	This work aims to analyze the triggering behaviors, focusing on the following two problems:
	\begin{problem}[Feasibility condition of Transition Relation]\label{prbfc}
    Given an event-triggered control system \eqref{closedloopsys} with triggering condition \eqref{eq:triggering-time}, determine necessary and sufficient conditions for the feasibility of IET transition relations from a geometric-algebraic perspective.
    \end{problem}	
    \begin{problem}[Algorithm For Triggering Behaviors]\label{prbtb}
     Compute the set of all feasible transitions $\mathcal{B}$.
    \end{problem}
The matrix $L$ can be viewed as the Laplacian matrix of a communication graph in a multi-agent system, where each agent updates its state based on neighboring information. In this setting, the studied problems correspond to understanding how event-triggered communication patterns emerge and how control parameters influence communication behaviors.

\section{A Geometric to Algebraic Conditions for Triggering Behavior Feasibility}
In this section, we characterize triggering behaviors from both geometric and algebraic perspectives. We begin by formulating the IET function, which provides a functional description of triggering. We then construct the model for triggering behaviors as a quadratic constraint satisfaction problem and reformulate it as a cone problem, yielding a geometric interpretation of feasibility. Finally, we establish an algebraic condition that serves as a computable condition for the feasibility of transition relations.
\subsection{IET Function Formulation}
We adopt the periodic event-triggering condition given by \cite{Tongwen2013Event}. At each sampling time, we check the triggering condition
	\begin{align}\label{triggercondi}
		\|e(\tau_k)\|^2 \leq \sigma \|Lx(\tau_k)\|^2,
	\end{align}
	where $k\in \mathbb{N}$, $e(\tau_k)=x(t_q)-x(\tau_k)$ is the error since the last control update, $ \sigma > 0 $ is a designed parameter, and $h$ is the sampling period. As shown in \cite{Tongwen2013Event}, the system achieves consensus if $0 < h \leq \frac{1}{2 \lambda_{n}}$ and $0 < \sigma <\frac{1}{ \lambda_{n}^2} $.

	It follows from \eqref{closedloopsys} that the state at $\tau_k$ can be given as
	\begin{align}
		x(\tau_k) = M(k)x(t_q), \label{evolutioin}
	\end{align}
	where $M(k)$ is called the transition matrix from $x(t_q)$ to $x(\tau_k)$, given by
	\begin{align}
		M(k) = I_n - khL. \label{Mk}
	\end{align}
	
	By noting that $e(\tau_k)=x(t_q)-x(\tau_k)$, we do the following formulation
	\begin{align*}
		\|e(\tau_k)\|^2 &\leq \sigma \|Lx(\tau_k)\|^2 \\
		\|x(t_q) - x(\tau_k)\|^2 &\leq \sigma \lambda_n^2 \|x(\tau_k)\|^2 \\
		(1-\sigma \lambda_n^2) x^T(\tau_k)x(\tau_k) &- 2x^T(t_q)x(\tau_k) + x^T(t_q)x(t_q) \leq 0 \ .
	\end{align*}
	To facilitate analysis, it is necessary to express the above inequality in the following quadratic form
	\begin{align}
	\begin{bmatrix} x(\tau_k) \\ x(t_q) \end{bmatrix}^T
	\begin{bmatrix} (1- \sigma  \lambda_{n}^2)I_n & -I_n \\ -I_n & I_n \end{bmatrix}
	\begin{bmatrix} x(\tau_k) \\ x(t_q) \end{bmatrix} \leq 0. \label{triggeredcondition}
	\end{align}

    With the above construction, the triggering condition \eqref{eq:triggering-time} is given by
	\begin{align}\label{triggertime}
		\begin{aligned}
				t_{q+1} = \inf \bigg\{&
			t = \tau_k, k \in \mathbb{N} \; \bigg| \;
			 \begin{bmatrix} x(\tau_k) \\ x(t_q) \end{bmatrix}^{\top}
			Q
			\begin{bmatrix} x(\tau_k) \\ x(t_q) \end{bmatrix} > 0 \\
			& \lor t - t_q \leq \bar{k}h \bigg\} ,
		\end{aligned}
	\end{align}
	where $\bar k$ is the chosen maximum IET, $\phi\big(x(t_q),x(\tau_k),\sigma\big) = \begin{bmatrix} x(\tau_k) \\ x(t_q) \end{bmatrix}^{\top}
			Q
			\begin{bmatrix} x(\tau_k) \\ x(t_q) \end{bmatrix}$, and $Q = \begin{bmatrix} (1- \sigma  \lambda_n^2)I_n & -I_n \\ -I_n & I_n \end{bmatrix} $.
	
	Finally, by combining \eqref{evolutioin} and \eqref{triggertime}, IET function \eqref{Kx} becomes
		\begin{align}\label{paraintereventtime}
		\begin{aligned}
			\mathcal{K}(x,\sigma)=\min  \big\{
			k \in \{1,...,\bar k\}
			\big |
			 x(t_q)^{\top}N(k)x(t_q)>0  \lor  k = \bar k\big\}
		\end{aligned}
	\end{align}
	where $N(k)$ is given by
	\begin{align}
		N(k)=\begin{bmatrix}M(k)\\I_n\end{bmatrix}^\top
		Q
		\begin{bmatrix}M(k)\\I_n\end{bmatrix}. \label{fdiNk}
	\end{align}
    
    \subsection{Representing Triggering Behaviors as a Cone Problem}
    
   According to Definition~\ref{Feasible}, the feasibility of a transition relation $k_1 \to k_2$ requires that there exists an initial state $x(t_q)$ such that the system first triggers at $k_1 h$ and subsequently at $k_2 h$. 
   This feasibility problem can be modeled as the following nonconvex quadratic constraint satisfaction problem (CSP) \cite{Manuel2020Scalable}
   \begin{equation} \label{nonconvex}
   	\begin{aligned}
   		&\exists x(t_q)\in\mathbb{R}^{n} \\
   		\text{s.t.} 
   		&x(t_q)^{\top}N(k_1)x(t_q) > 0, \\
   		&x(t_q)^{\top}N(j)x(t_q)\leq0,\forall j\in\{1,...,k_1-1\}, \\
   		&x(t_q)^{\top}M(k_{1})^{\top}N(k_{2})M(k_{1})x(t_q) > 0, \\
   		&x(t_q)^{\top}M(k_{1})^{\top}N(j)M(k_{1})x(t_q)\leq0,\\
   		&\forall j\in\{1,...,k_{2}-1\}.
   	\end{aligned}
   \end{equation}
    
    In this formulation, the existence of the transition $k_1 \to k_2$ is fully characterized by the existence of a state $x(t_q)$ satisfying the above constraints. If such a state exists, the transition relation is said to be feasible; otherwise, it is considered infeasible.
   \begin{remark}
       The core idea of formulation~\eqref{nonconvex} is derived through three steps. First, the system dynamics are decoupled from the state trajectory by using the state transition matrix $M(k)$, so that the future evolution can be expressed algebraically without explicitly simulating time.  Second, the requirement that the first trigger occurs exactly at $k_1 h$ is translated into the constraint set $x(t_q)^\top N(j)x(t_q)\le 0$ for all $j<k_1$ together with $x(t_q)^\top N(k_1)x(t_q)>0$. Third, combining the algebraic evolution with the same logic for the next trigger at $k_2 h$ yields a set of static quadratic inequalities on $x(t_q)$. In this way, the feasibility of the transition relation $k_1 \to k_2$ is transformed into a constraint satisfaction problem.
   \end{remark}
   
    \begin{remark}
    The CSP~\eqref{nonconvex} is challenging because its 
    constraints take quadratic forms such as $x^\top N(k)x > 0$ and 
    $x^\top M(k_1)^\top N(k_2) M(k_1)x > 0$. These inequalities are nonlinear, 
    since their boundaries are quadratic surfaces rather than hyperplanes. 
    Moreover, the matrices $N(k)$ and $M(k_1)^\top N(k_2)M(k_1)$ are indefinite in general. As a result, the feasible sets are nonconvex and bounded by curved regions, whose intersections produce complicated geometries that are difficult to analyze directly.
    \end{remark}

    To overcome these limitations, we convert the original quadratic constraints problem into a linear cone problem~\eqref{eqLP}, whose feasibility of transition relation $k_1 \to k_2$ is equivalent to \eqref{nonconvex}.
    \begin{lemma}\label{lemLP}
    	The transition relation $k_1 \to k_2$ is feasible 
    	if and only if the following linear cone problem admits a solution
    	\begin{equation}\label{eqLP}
    		\begin{aligned}
    			&\exists\,z(t_q)\in\mathbb{R}^{n} \\
    			\text{s.t.}\quad
    			&-d(k_1,\sigma)^\top z(t_q) < 0, \\
    			&\phantom{-}d(j,\sigma)^\top z(t_q)\leq 0,\quad \forall j\in\{1,...,k_1-1\}, \\
    			&-t(k_1,k_2,\sigma)^\top z(t_q) < 0, \\
    			&\phantom{-}t(k_1,j,\sigma)^\top z(t_q)\leq 0,\quad \forall j\in\{1,...,k_{2}-1\}, \\
    			&-e_i z(t_q) \leq 0, \quad \forall i \in \{1,...,n\} .
    		\end{aligned}
    	\end{equation}
    where $z(t_q) = (V^{\top}x(t_q)) \circ (V^{\top}x(t_q))$, $d(k,\sigma) = \operatorname{diag}(D(k))$ and $t(k_1,k_2,\sigma) = \operatorname{diag}(T(k_1,k_2))$, with $D(k) = V^{\top}N(k)V$, $T(k_1,k_2) = V^{\top}M(k_1)^{\top}N(k_2)M(k_1)V$ being diagonal matrices, and $e_i$ denotes the $i$-th standard basis vector in $\mathbb{R}^n$, i.e., $e_i = [0,\dots,1,\dots,0]^\top$ with $1$ at the $i$-th entry. In other words, systems~\eqref{nonconvex} and~\eqref{eqLP} are equivalent in feasibility.
    \end{lemma}
    \begin{remark}
    Unlike the SDR approach~\cite{Manuel2020Scalable}, which lifts the quadratic constraints to $X=xx^\top$ and replaces strict inequalities with non-strict ones while dropping the rank-one condition, our method reformulates \eqref{nonconvex} exactly as a linear cone problem. This preserves the original strict inequalities and ensures full equivalence with the quadratic formulation. As a result, our approach avoids the conservatism inherent in SDR and provides a clearer geometric description of the feasible region, making subsequent analysis more reliable. 
\end{remark}
\begin{figure}[htpb]
    	\centering
    	\includegraphics[width=1.5in]{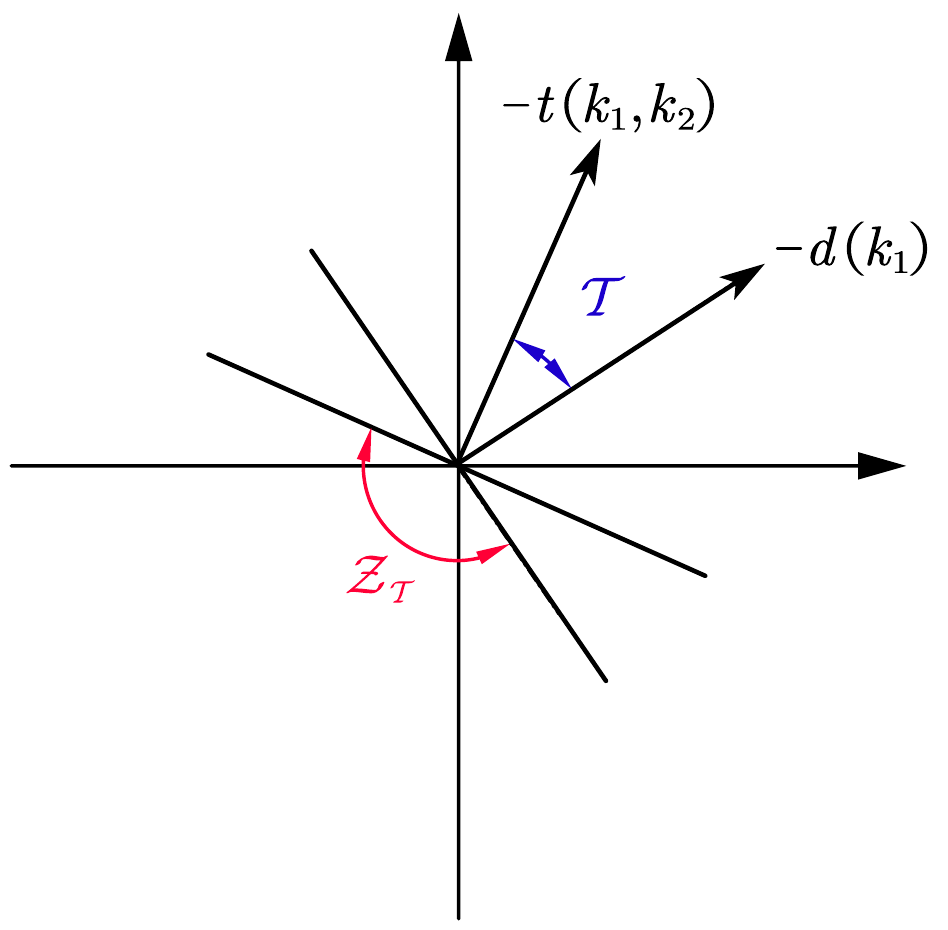}
    	\caption{Visualization of the Cone $\mathcal{T}$ and Cone $\mathcal{Z}_\mathcal{T}$}
    	\label{ConeVisual}
    \end{figure}
\begin{remark}
      To illustrate why Lemma~\ref{lemLP} leads to a cone problem, consider the example in Fig.~\ref{ConeVisual}. The two strict inequality constraints $-d(k_1)^\top z(t_q)<0$ and $-t(k_1,k_2)^\top z(t_q)<0$ generate an open cone $\mathcal{T}$ spanned by the vectors $-d(k_1)$ and $-t(k_1,k_2)$. The feasible states $z(t_q)$ must lie in the polar cone $\mathcal{Z}_{\mathcal{T}}$, i.e., all directions forming strictly negative inner products with vectors in $\mathcal{T}$. 
  \end{remark}

\subsection{A Computable Condition for the Feasibility of Triggering Behaviors: A Geometric-Algebraic Perspective}
   
    Based on Lemma \ref{lemLP}, we now establish a rigorous 
    criterion that unifies the geometric and algebraic perspectives. This leads to the following theorem, which provides necessary and sufficient conditions for the feasibility of $k_1 \to k_2$.
   
   \begin{theorem}[Feasibility condition of Transition Relation]\label{thm:dual-full}
   	The transition relation $k_1 \to k_2$ is \textbf{infeasible} if and only if there exists non-negative multipliers $s_1$, $s_2$, $\{\mu_{i}\}_{i=1}^{k_1-1}$, $\{\nu_{i}\}_{i=1}^{k_2-1}$, $\{\eta_i\}_{i=1}^{n}$, with  $s_1 + s_2 > 0$ such that
   	\begin{equation}\label{dualequality}
   		\begin{aligned}
   			&-s_1\,d(k_{1},\sigma)
   			-s_2\,t(k_{1},k_{2},\sigma)\\
   			&+\sum_{i=1}^{k_{1}-1}\mu_{i}\,d(i,\sigma)
   			+\sum_{i=1}^{k_{2}-1}\nu_{i}\,t(k_{1},i,\sigma)
   			-\sum_{i=1}^{n}\eta_i\,e_i = 0.
   		\end{aligned}
   	\end{equation} 
   	where $d(k,\sigma)=\text{diag}(D(k)) \text{ and } t(k_1,k_2,\sigma)=\text{diag}(T(k_1,$ $k_2))$.
   \end{theorem}
   
   \begin{proof}\textbf{Necessity.} The set of states $ z(t_q)$ 
   	that satisfy the strict triggering constraints $-d(k_1)^\top z(t_q) < 0$ and $-t(k_1,k_2)^\top z(t_q) < 0$ is given by
   	\begin{align}\label{Zt}
   		\mathcal{Z}_{\mathcal{T}} := 
   		\left\{ z \in \mathbb{R}^n \,\middle|\,
   		\langle z, v \rangle < 0,\ \forall v \in \mathcal{T}
   		\right\},
   	\end{align}
   	where $\mathcal{T} \in \mathbb{R}^n$ denotes the open convex cone generated by the two strict triggering constraints, defined as
   	\begin{align}\label{T}
   		\mathcal{T} := \left\{ -s_1\,d(k_1,\sigma) - s_2\,t(k_1,k_2,\sigma)\ \middle| s_1, s_2 > 0 \right\}.
   	\end{align}
    with positive scalars $s_1$ and $s_2$.
    
    Similarly, we define the set of states that satisfy all non-strict inequality constraints as
   	\begin{align}
   		\mathcal{Z}_{\mathcal{C}} := 
   		\left\{ z \in \mathbb{R}^n \,\middle|\,
   		\langle z, u \rangle \leq 0,\ \forall u \in \mathcal{C}
   		\right\}
   	\end{align}
   	where $\mathcal{C} \in \mathbb{R}^n$ denotes the closed convex cone generated by all non-strict constraints
   	\begin{equation}\label{Cset}
   		   	\begin{aligned}
   			\mathcal{C} := \Big\{ 
   			\sum_{i=1}^{k_{1}-1}\mu_{i}\,d(i,\sigma)
   			+\sum_{i=1}^{k_{2}-1}\nu_{i}\,t(k_{1},i,\sigma)\\
   			-\sum_{i=1}^{n}\eta_i\,e_i
   			\ \Big|\
   			\mu_{i}, \nu_{i}, \eta_i \geq 0
   			\Big\},
   		\end{aligned}
   	\end{equation}
   	and $e_i$ denotes the $i$-th standard basis vector in $\mathbb{R}^n$, i.e., $e_i = [0,\dots,1,\dots,0]^\top$ with $1$ at the $i$-th entry. The non-negative scalars $\mu_{k_1'}, \nu_{k_2'}, \eta_i$ are conic multipliers associated with non-strict triggering constraints.
   	
   	If~\eqref{eqLP} were feasible, then there would exist a vector $z(t_q)$ such that all inner products with vectors in $\mathcal C$ are non-positive, and with those in $\mathcal T$ strictly negative. This implies that the intersection of the corresponding feasible regions is non-empty, i.e., $\mathcal{Z}_{\mathcal{T}} \cap \mathcal{Z}_{\mathcal{C}} \neq \emptyset$.  
   	Conversely, if~\eqref{eqLP} is infeasible, we must have
    \begin{equation} \label{g2acond}
        \mathcal{Z}_{\mathcal{T}} \cap \mathcal{Z}_{\mathcal{C}} = \emptyset. 
    \end{equation}
    Since $\mathcal{Z}_{\mathcal{C}}$ is a closed convex cone and $\mathcal{Z}_{\mathcal{T}}$ is an open convex cone, which guarantees the existence of a separating hyperplane \cite[Thm~2.4]{kesavan2017hahn}. In our setting, we take the normed space to be $V = \mathbb{R}^n$, with $A = \mathcal{Z}_{\mathcal{T}}$ (the open convex cone) and $B = \mathcal{Z}_{\mathcal{C}}$ (the closed convex cone).
   	Since $\mathbb{R}^n$ is finite-dimensional, its topological dual space $V^* = (\mathbb{R}^n)^*$ is isomorphic to $\mathbb{R}^n$. Consequently, any continuous linear functional $f \in V^*$ can be written as an inner product with a fixed vector $\theta \in \mathbb{R}^n$, i.e., $f(x) = \langle \theta, x \rangle$ for all $x \in \mathbb{R}^n$.
   	  		  	
   	Evaluating the separating inequality at $u = 0 \in \mathcal{Z}_{\mathcal{C}}$ yields $\alpha \leq \langle \theta, 0 \rangle = 0$, implying $\alpha \leq 0$. On the other hand, although $\mathcal{Z}_{\mathcal{T}}$ is an open cone and thus does not contain the origin, its closure $\operatorname{cl}(\mathcal{Z}_{\mathcal{T}})$ does. For any $v \in \mathcal{Z}_{\mathcal{T}}$, the scaled vector $\lambda v$ belongs to $\mathcal{Z}_{\mathcal{T}}$ for all $\lambda > 0$, and $\lim_{\lambda \to 0^+} \lambda v = 0 \in \operatorname{cl}(\mathcal{Z}_{\mathcal{T}})$. Taking the limit in the inequality $\langle \theta, \lambda v \rangle \le \alpha$ yields $\lim_{\lambda \to 0^+} \langle \theta, \lambda v \rangle = 0$,
	which implies $\alpha \ge 0$.
   	
   	Combining both bounds, $\alpha \le 0$ and $\alpha \ge 0$, we conclude that $\alpha = 0$. Hence, the separation condition $f(x) \le \alpha \le f(y),\, \forall x \in A,\ y \in B$ becomes
   	\begin{align}\label{sepration}
   		\langle \theta, u \rangle \geq 0 \quad \forall u \in \mathcal{Z}_{\mathcal{C}}, \qquad
   		\langle \theta, v \rangle \leq 0 \quad \forall v \in \mathcal{Z}_{\mathcal{T}}.
   	\end{align}

    	Assume $ \mathcal{Z}_{\mathcal{T}} \cap \mathcal{Z}_{\mathcal{C}} = \emptyset$, there exists $\theta \setminus \{0\}$ satisfying the above. From the separation inequality $\langle \theta, u \rangle \geq 0, \,  \forall u \in \mathcal{Z}_{\mathcal{C}}$ in \eqref{sepration}, we have $\theta \in \mathcal{Z}_{\mathcal{C}}^*$. Since $\mathcal{Z}_{\mathcal{C}}=C^\circ$ and $\mathcal C$ is a closed convex cone, we have $\mathcal{Z}_{\mathcal{C}}^*=(C^\circ)^*=-C$. Therefor, $\theta \in -C$, which can be written as
    \begin{equation}\label{eq:neg_theta_derivation1}
    	   	\begin{aligned}
    		\theta \in -\Big\{ 
    		\sum_{i=1}^{k_{1}-1}\mu_{i}\,d(i,\sigma)
   			+\sum_{i=1}^{k_{2}-1}\nu_{i}\,t(k_{1},i,\sigma)\\
   			-\sum_{i=1}^{n}\eta_i\,e_i
   			\ \Big|\
   			\mu_{i}, \nu_{i}, \eta_i \geq 0
    		\Big\}
    	\end{aligned}
    \end{equation}
   	
   	Similarly, from the separation inequality $\langle \theta, v \rangle \leq 0 ,\, \forall v \in \mathcal{Z}_{\mathcal{T}}$, we equivalently have $\langle -\theta, v \rangle \geq 0 ,\, \forall v \in \mathcal{Z}_{\mathcal{T}}$, which implies that $-\theta \in \mathcal{Z}_{\mathcal{T}}^*$. Since $\mathcal{Z}_{\mathcal{T}}$ is open convex set and $\mathcal{Z}_{\mathcal{T}}=\operatorname{int}(\mathcal{T}^\circ)$, its dual cone $\mathcal{Z}_{\mathcal{T}}^*=(\operatorname{int}(\mathcal{T}^\circ))^*=-\operatorname{cl}(\mathcal{T})$ is a closed convex cone. The condition $-\theta \in -\operatorname{cl}(\mathcal{T})$ can thus be expressed algebraically as:
   	\begin{equation}\label{eq:neg_theta_derivation}
   		-\theta \in -\{s_1 (-d(k_1,\sigma)) + s_2 (-t(k_1,k_2,\sigma)) \ \Big|\ s_1, s_2 \ge 0 \}.
   	\end{equation}
   	 
   	If $s_1=s_2=0$ in \eqref{eq:neg_theta_derivation}, then $v=0 \in \mathcal{T}$, which violates the strict inequality
   	$\langle z,v\rangle<0,\ \forall v \in \mathcal{T}$ holds by \eqref{Zt}.  Hence at least one of $s_1,s_2$ must be
   	strictly positive, implying that
   	\begin{align}\label{s1s2}
   		s_1+s_2>0.
   	\end{align}
   	
   	Combining \eqref{eq:neg_theta_derivation1} and \eqref{eq:neg_theta_derivation} yields $s_1$, $s_2$, $\{\mu_{i}\}_{i=1}^{k_1-1}$, $\{\nu_{i}\}_{i=1}^{k_2-1}$, $\{\eta_i\}_{i=1}^{n}$ with all components non-negative and
   	$s_1+s_2>0$ satisfying the vector identity
   	\begin{equation}\label{mutp}
   		\begin{aligned}
   			 &-s_1\,d(k_{1},\sigma)
   			-s_2\,t(k_{1},k_{2},\sigma)\\
   			&+\sum_{i=1}^{k_{1}-1}\mu_{i}\,d(i,\sigma)
   			+\sum_{i=1}^{k_{2}-1}\nu_{i}\,t(k_{1},i,\sigma)
   			-\sum_{i=1}^{n}\eta_i\,e_i = 0,
   		\end{aligned}
   	\end{equation}
   	thereby completing the necessary part of the theorem.	 	
   	
   	\textbf{Sufficiency.} Suppose that there exists a non-negative multipliers $s_1$, $s_2$, $\{\mu_{i}\}_{i=1}^{k_1-1}$, $\{\nu_{i}\}_{i=1}^{k_2-1}$, $\{\eta_i\}_{i=1}^{n}$, with $s_1 + s_2 > 0$, such that the vector identity \eqref{mutp} holds. For contradiction, assume that the system \eqref{eqLP} is also feasible, which means there exists a vector $z$ satisfying all inequalities.
   	
   	We take the inner product of both sides of the identity \eqref{mutp} with the solution $z$ yields
   	$\bigl\langle z,\,-s_1\,d(k_{1},\sigma)-s_2\,t(k_{1},k_{2},\sigma)+\sum_{i=1}^{k_{1}-1}\mu_{i}\,d(i,\sigma)
   	+\sum_{i=1}^{k_{2}-1}\nu_{i}\,t(k_{1},i,\sigma)-\sum_{i=1}^{n}\eta_i\,e_i\bigr\rangle=0.$ Since $s_1+s_2>0$,
   	the whole expression is strictly negative,
   	contradicting equality to zero.  
   	Hence, no such $z$ can exist and \eqref{eqLP} is infeasible.
   	
   	This contradiction shows that our assumption that both system \eqref{eqLP} and system \eqref{dualequality} can be feasible simultaneously must be false.
   	Thus, if \eqref{dualequality} is feasible, system \eqref{eqLP} must be infeasible. This completes the sufficiency part.
   	\end{proof}  

    \begin{remark}
         The geometric condition $\mathcal{Z}_{\mathcal{T}} \cap \mathcal{Z}_{\mathcal{C}} = \emptyset$ defined by \eqref{g2acond},  shows that the feasibility of a transition relation $k_1 \to k_2$ is equivalent to the existence of an intersection between $\mathcal{Z}_{\mathcal{C}}$ and $\mathcal{Z}_{\mathcal{T}}$. Fig.~\ref{ConeFeas} further visualizes the relation between $\mathcal{Z}_\mathcal{C}$ and $\mathcal{Z}_\mathcal{T}$. In the left panel, the two cones do not intersect, so the transition $k_1 \to k_2$ is infeasible.  In the right panel, the cones intersect, meaning that a feasible state $z(t_q)$ exists and the transition $k_1 \to k_2$ is feasible, and the overlapping region corresponds to the feasible set of $z(t_q)$.
    \end{remark}
    \begin{figure}[htpb]
    	\centering
    	\includegraphics[width=3.5in]{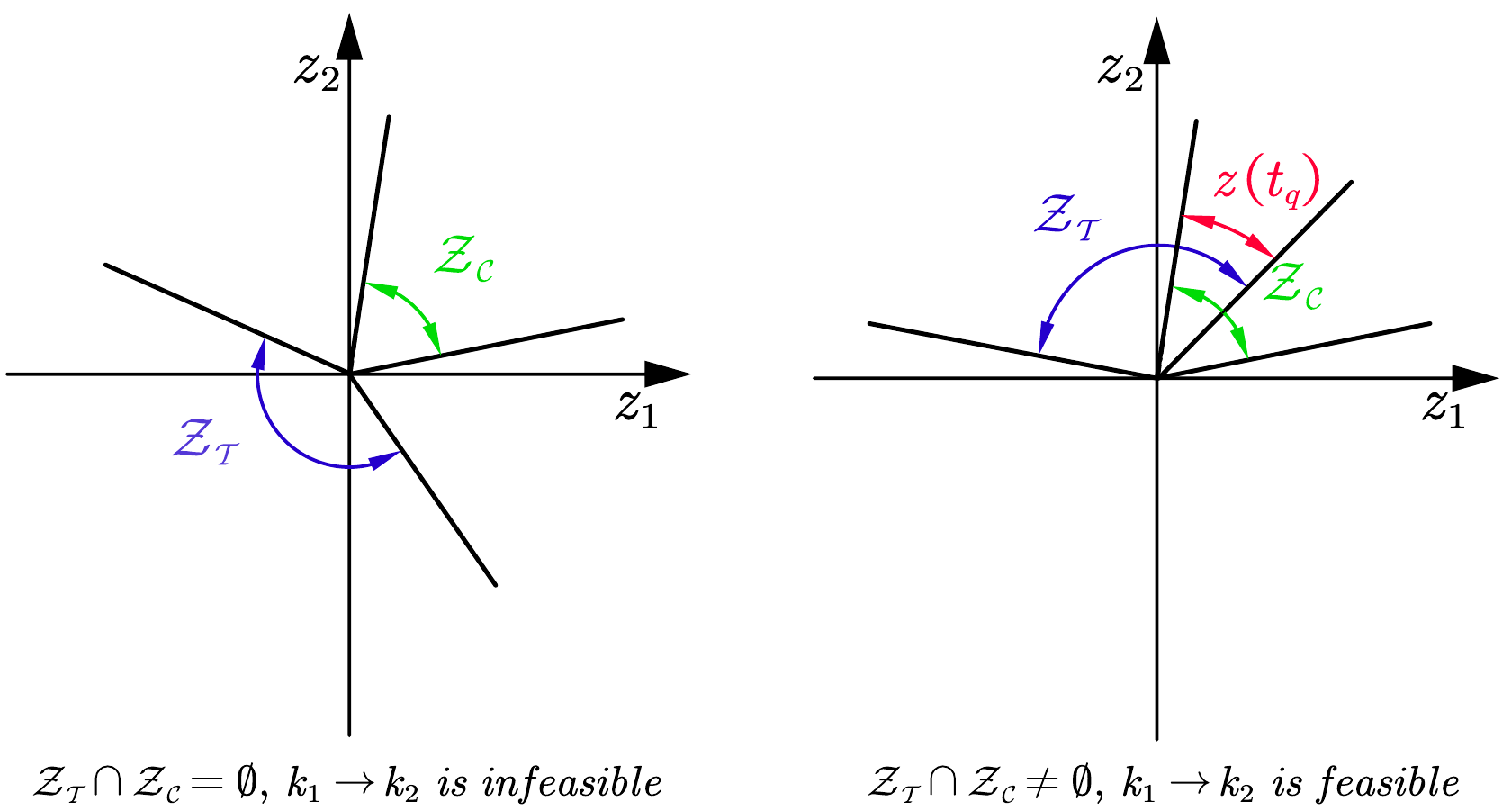}
    	\caption{Geometric visualization of Feasibility Condition}
    	\label{ConeFeas}
    \end{figure}
    
    \begin{remark}
    Fig.~\ref{Seperat} illustrates the connection between the geometric and algebraic views of infeasibility. On the left, the cones $\mathcal{Z}_{\mathcal{C}}$ and $\mathcal{Z}_{\mathcal{T}}$ are disjoint, which geometrically certifies that the transition relation is infeasible. On the right, this separation is expressed
    algebraically by a hyperplane, whose normal vector corresponds to nonnegative multipliers satisfying \eqref{dualequality}. Thus, the figure makes explicit how the existence of such multipliers provides the algebraic counterpart to the geometric separation condition.
    \end{remark}

    \begin{figure}[htpb]
    	\centering
    	\includegraphics[width=3.5in]{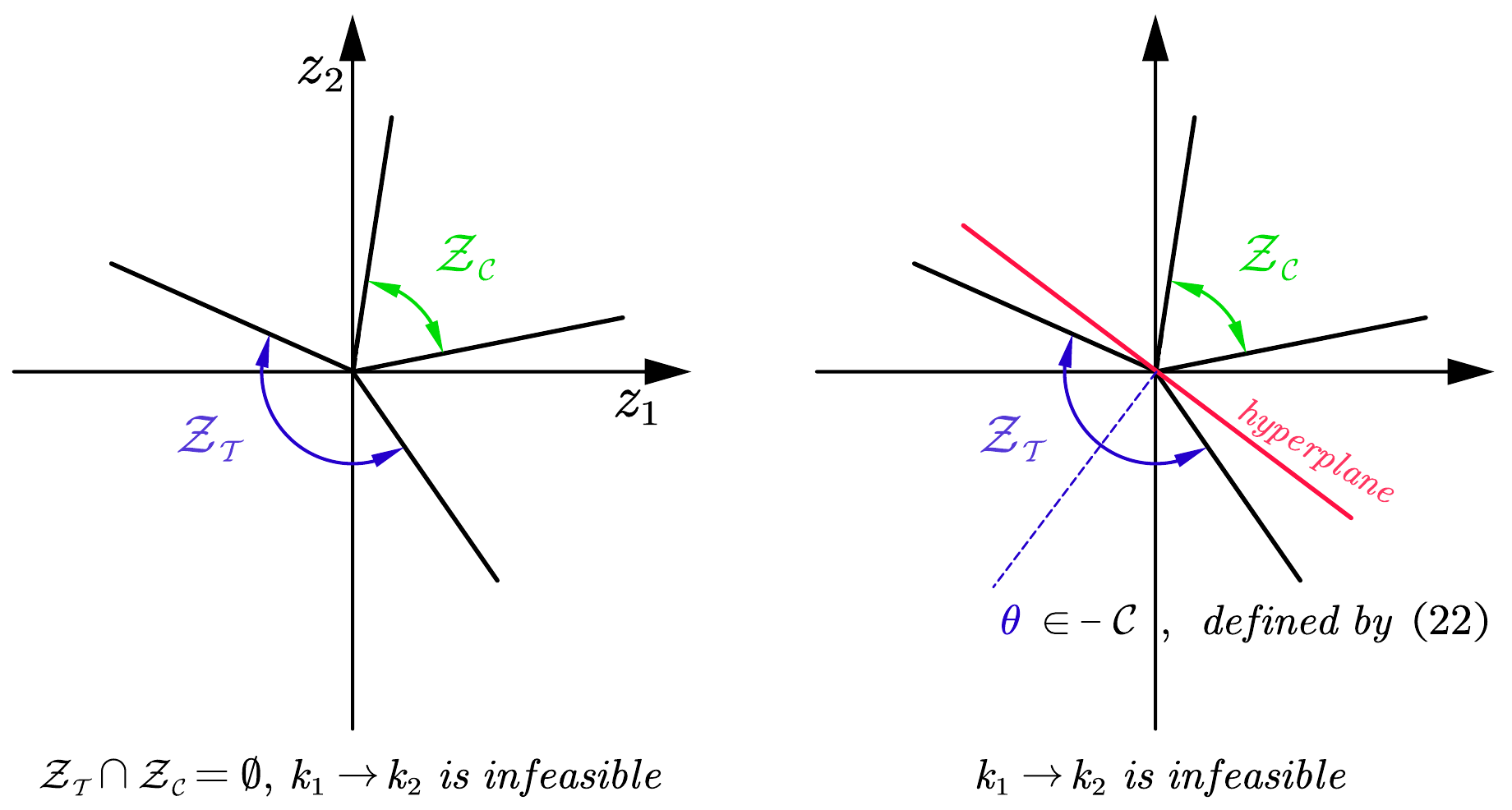}
    	\caption{Geometric-to-Algebraic Transformation of Feasibility Condition}
    	\label{Seperat}
    \end{figure}

    \begin{remark}
        While Theorem~\ref{thm:dual-full} provides a geometric condition for feasibility, treating the control parameter~$\sigma$ as a variable turns this condition into a powerful analytical tool. The geometric perspective reveals how the cones $\mathcal{Z}_\mathcal{C}$ and $\mathcal{Z}_\mathcal{T}$ vary as $\sigma$ varies, thereby offering direct insight into how feasibility regions of triggering behaviors emerge or disappear.
    \end{remark}

To identify all feasible triggering behaviors, we construct the candidate transition relations $\mathcal{D}$ defined in \eqref{cl} and retain only those that satisfy the condition in~\eqref{dualequality}. The procedure is summarized in Algorithm~\ref{alg:offline} and {\textit{\bf Problem 2}} is solved.
    \begin{algorithm}[!h]
            \caption{Construction of $\mathcal{B}$}
            \label{alg:offline}
            \renewcommand{\algorithmicrequire}{\textbf{Input:}}
            \renewcommand{\algorithmicensure}{\textbf{Output:}}
            
            \begin{algorithmic}[1]
            \REQUIRE $h$, $\bar k$ , $\mathcal{D}$ 
            \ENSURE $\mathcal{B}$
            
            \STATE Generate the candidate sequences sets $\mathcal{D}$ using \eqref{cl}
            \STATE Initialize the IET transition relations set: $\mathcal{B} \gets \emptyset$
            
            \FOR{each candidate $(k_1, k_2) \in \mathcal{D}$}
                \IF{the linear cone problem \eqref{dualequality} has no solution associated with $(k_1 , k_2)$}
                    \STATE $\mathcal{B} \gets \mathcal{B} \cup \{(k_1, k_2)\}$
                \ENDIF
            \ENDFOR
    
        \RETURN $\mathcal{B}$
        \end{algorithmic}
    \end{algorithm}

\section{Numerical Results}	
 Consider the system described in~\eqref{eq:linearSys} under the event-triggering condition given by~\eqref{triggertime}. While our theoretical results hold for a general N-agent system, visualizing the impact of the control parameter $\sigma$ on the feasibility of triggering behaviors is most instructive in a low-dimensional space. To this end, we select a two-agent system, whose communication topology is captured by the corresponding Laplacian matrix $L=\begin{bmatrix}
     1&-1\\
    -1& 1
 \end{bmatrix}$. This setup reduces the dimensionality of the problem to $n=2$, enabling a direct geometric interpretation and visualization of how the cones $\mathcal{Z}_{\mathcal{C}}$ and $\mathcal{Z}_{\mathcal{T}}$ vary as $\sigma$ is varied. Then, we have $\lambda_n = 2$, which yields the admissible parameter ranges $h \in (0, 0.25]$ and $\sigma \in (0, 0.25]$ according to the stability condition $0 < h \leq \frac{1}{2 \lambda_{n}}$ and $0 < \sigma <\frac{1}{ \lambda_{n}^2} $ given by \cite{Tongwen2013Event}. 

 We consider the case with sampling period $h=0.06$, control parameter $\sigma=0.0469$, and maximum IET $\bar{k}=20$. Using Algorithm~\ref{alg:offline}, all feasible triggering behaviors are obtained, as illustrated in Fig.~\ref{Behaviors}(a). In the figure, the horizontal axis represents the current triggering IET, while the vertical axis indicates the possible IETs for the next trigger. For example, if the current triggering interval is $5h$, then the next IET may take values $9h$, $10h$, $11h$, $12h$, or $13h$, corresponding to the transition relations $5 \to 9$, $5 \to 10$, $5 \to 11$, $5 \to 12$, $5 \to 13$.

   \begin{figure}[thpb]
	\centering
	\begin{minipage}{0.23\textwidth}
		\centering
		\includegraphics[width=1.6in]{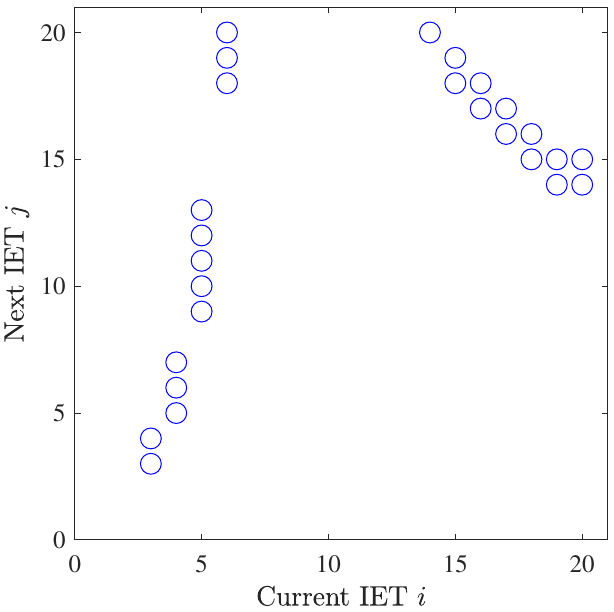}
		\text{(a) $\sigma=0.0469$. }
	\end{minipage}
	\begin{minipage}{0.23\textwidth}
		\centering
		\includegraphics[width=1.6in]{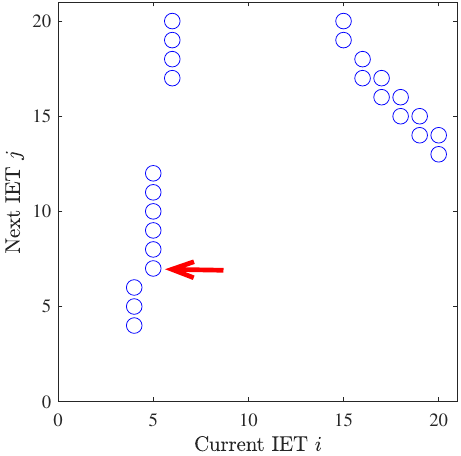}
		\text{(b) $\sigma=0.1314$.}
	\end{minipage}
	\caption{Triggering Behaviors with different $\sigma$.}
	\label{Behaviors}
\end{figure}

It can be observed that at $\sigma=0.0469$ the transition relation $5 \to 7$ does not exist. To further examine how this transition relation is affected by $\sigma$, and how $\sigma$ can be tuned to enable or disable it, we exploit the geometric condition \eqref{g2acond} provided in Theorem~\ref{thm:dual-full}. Specifically, the feasibility of $5 \to 7$ is determined by whether the intersection $\mathcal{Z}_{\mathcal{T}} \cap \mathcal{Z}_{\mathcal{C}}$ is nonempty, and the evolution of this intersection with respect to $\sigma$ reveals the behavior. 
To illustrate this, we select two values of $\sigma$ 
($0.0469$, and $0.1314$) and plot the corresponding feasible state space, as shown in Fig.~\ref{ConeVisualCap}. To further validate this result, we apply Algorithm~\ref{alg:offline} again with $\sigma=0.1314$ to construct all triggering behaviors, as shown in Fig.~\ref{Behaviors}(b), where the transition $5 \to 7$ indeed appears, which is highlighted by the red arrow.

 	\begin{figure}[thpb]
	\centering
	\begin{minipage}{0.23\textwidth}
		\centering
		\includegraphics[width=1.6in]{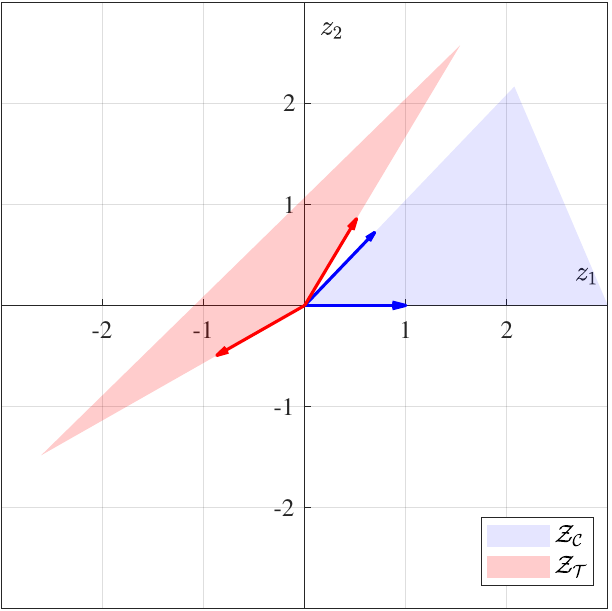}
		\text{(a) $\sigma=0.0469$.}
	\end{minipage}
    \begin{minipage}{0.23\textwidth}
		\centering
		\includegraphics[width=1.6in]{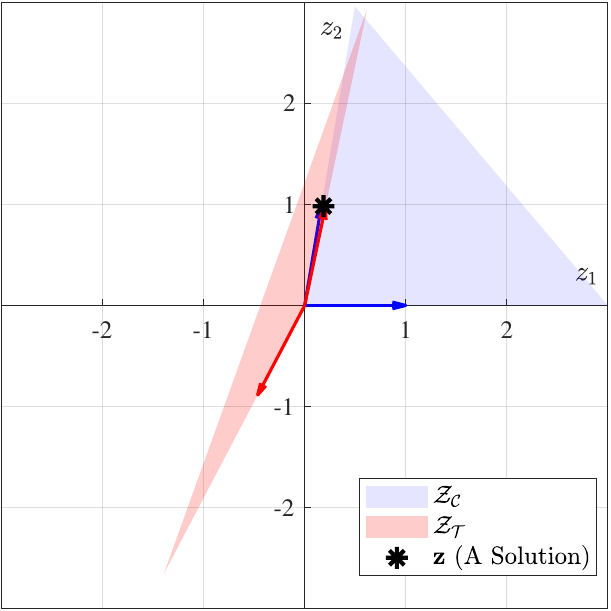}
		\text{(b) $\sigma=0.1314$.}
	\end{minipage}
	\caption{Feasible state space of transition relation $5 \to 7$.}
	\label{ConeVisualCap}
\end{figure}

Fig.~\ref{ConeVisualCap} shows how the feasibility of the transition relation $5 \to 7$ depends on $\sigma$. For $\sigma=0.0469$, the cones $\mathcal{Z}_{\mathcal{T}}$ and $\mathcal{Z}_{\mathcal{C}}$ are disjoint and the feasible state space does not exist, so the transition is infeasible. At $\sigma=0.1314$, the cones intersect, and their overlapping region defines the feasible state space, indicating that the transition relation is feasible. This observation is consistent with the simulation results in Fig.~\ref{Behaviors}(b), and highlights the critical role of $\sigma$ in enabling or disabling a specific transition relation.

\section{Conclusion}
In this paper, we characterize the triggering behaviors of event-triggered 
control systems from a geometric–algebraic perspective. We first show that the feasibility of an IET transition relation can be formulated as a nonconvex quadratic CSP. This problem is then reformulated as an equivalent linear cone problem, allowing the feasibility of the transition relation to be captured by the intersection of cones. Building on this result, we establish an equivalent algebraic condition for the existence of any IET transition relation, and develop an algorithm to compute the complete set of all feasible transition relations. Numerical simulations confirm that the proposed method not only matches observed triggering behaviors but also offers an intuitive tool for parameter analysis. This work thus lays a rigorous foundation for the exact characterization and design of triggering behaviors in event-triggered control systems.

\bibliographystyle{IEEEtran}
\bibliography{refs.bib}

\end{document}